\documentclass[letterpaper, 10 pt, conference]{ieeeconf}  

\IEEEoverridecommandlockouts                              
\usepackage{graphics} 
\usepackage{amsmath} 
\usepackage{amssymb}  

\newtheorem{proposition}{Proposition}
\newtheorem{theorem}{Theorem}

\title{\LARGE \bf
Optimally Designing Cybersecurity Insurance Contracts to Encourage the Sharing of Medical Data
}

\author{Yoon Lee and Anil Aswani
\thanks{This material is based upon work partially supported by the National Science Foundation under Grant CMMI-1847666, and partially supported by the UC Berkeley Center for Long-Term Cybersecurity.}
\thanks{Y. Lee and A. Aswani are with the Department of Industrial Engineering and Operations Research,
        University of California, Berkeley, CA 94720, USA
        {\tt\small yllee@berkeley.edu, aaswani@berkeley.edu}}%
}

\begin{document}

\maketitle
\thispagestyle{empty}
\pagestyle{empty}

\begin{abstract}

Though the sharing of medical data has the potential to lead to breakthroughs in health care, the sharing process itself exposes patients and health care providers to various risks. Patients face risks due to the possible loss in privacy or livelihood that can occur when medical data is stolen or used in non-permitted ways, whereas health care providers face risks due to the associated liability. For medical data, these risks persist even after anonymizing/deidentifying, according to the standards defined in existing legislation, the data sets prior to sharing, because shared medical data can often be deanonymized/reidentified using advanced artificial intelligence and machine learning methodologies. As a result, health care providers are hesitant to share medical data. One possible solution to encourage health care providers to responsibly share data is through the use of cybersecurity insurance contracts. This paper studies the problem of designing optimal cybersecurity insurance contracts, with the goal of encouraging the sharing of the medical data. We use a principal-agent model with moral hazard to model various scenarios, derive the optimal contract, discuss its implications, and perform numerical case studies. In particular, we consider two scenarios: the first scenario is where a health care provider is selling medical data to a technology firm who is developing an artificial intelligence algorithm using the shared data. The second scenario is where a group of health care providers share health data amongst themselves for the purpose of furthering medical research using the aggregated medical data.

\end{abstract}

\section{Introduction}

The rapid development of new artificial intelligence algorithms for health care has the potential to lead to an era of computational precision health \cite{roski2014, lee2017, mintz2017behavioral, zhou2018evaluating, cdc2018, aha2019, hulsen2020}. The development of these algorithms requires access to large sets of medical data. Nonetheless, the sharing of such medical data poses risks to patients due to the possible loss in privacy or livelihood that can occur when medical data is stolen or used in non-permitted ways. New ideas for the cybersecurity of medical data are needed to ensure that these new advances can continue to be developed.

\subsection{Privacy Risks from Sharing Medical Data}

A unique aspect of medical data is that even when it has been anonymized/deidentified \cite{kanonymity,ldiversity,tcloseness,dprivacy}  (in accordance with legislation like HIPAA \cite{hipaa} or GDPR \cite{gdpr}) prior to sharing, the data can often be deanonymized/reidentified \cite{narayanan2008, ohm2010, datta2012, oxford2022}. Examples include  deanonymization of a Massachusetts hospital database by joining it with a public voter database \cite{sweeney1997} and reidentification of a physical activity data set from the National Health and Nutrition Examination Survey (NHANES) using standard machine learning \cite{na2018}. 

In addition, a recent study has revealed that more than two-thirds of hospital data breaches include sensitive demographic and financial information that could lead not only to fraud and identity theft but also to discrimination and violation of fundamental rights \cite{jiang2020}. This highlights the necessity of developing approaches to safeguard patients and health care providers against cybersecurity threats.

\subsection{Cybersecurity of Medical Data}

The above described privacy risks deter health care providers from sharing their data \cite{panhuis2014, pisani2010}. One possible approach to mitigating some of the risks with sharing health data is through the design of cybersecurity insurance contracts. For instance, cybersecurity insurance can be used to partially compensate for the costs involved with recovery from a cyber-related security breach or similar incidents \cite{marotta2017}. 

A growing literature studies cybersecurity insurance. For instance, \cite{kunreuther2003, ogut2005, bolot2008} focus on the interdependent security problem to verify whether firms have adequate incentives to invest in protection against a risk whose magnitude depends on the action of others. The work in \cite{bohme2005, bohme2006} introduces new models and measures for the correlation of cyber-risks within and across independent firms, while \cite{shetty2010} investigates the issue of information asymmetries, namely in the form of moral hazard, when cyber-insurers cannot observe individual user security levels. The studies \cite{bohme2010, schwartz2014} provide a unifying framework to address the aforementioned hurdles that complicate risk management via cyber-insurance. 

\subsection{Contributions and Outline}

In this paper, we study the problem of incorporating cybersecurity insurance, which has been mainly explored in the setting of interdependent and correlated networks, into the design of contracts that govern the sharing of medical data. Such contracts would not only protect health care providers against losses resulting from a cyber-attack, but have the potential to foster the sharing of medical data.

In Sect. \ref{sec:scenarioa}, we analyze the scenario in which a health care provider sells medical data to a technology firm that uses the data to develop new artificial intelligence algorithms. We provide mathematical models for both parties, formulate a contract design problem in the setting of a principal-agent model with moral hazard \cite{laffont2002}, derive the optimal contract, and discuss insights gained from the optimal contract. In Sect. \ref{sec:scenariob}, we analyze a second scenario in which a group of health care providers forms a consortium to share medical data with each other for the purpose of conducting scientific research and improving patient care. Again, we provide a mathematical model for the health care providers, formulate a contract design problem, derive the optimal contract, and discuss insights gained from the optimal contract.


\section{Scenario A: Health Care Provider Selling Medical Data to Technology Firm}
\label{sec:scenarioa}

The first scenario we study is that of a health care provider selling medical data to a technology firm that is developing artificial inteligence algorithms using the shared data. Here, an important consideration to the health care provider is the quality of cybersecurity that the technology firm uses to protect any medical data they receive. If the firm suffers from a data breach, then the health care provider itself will face liability from those patients whose medical data has been breached. Thus, the health care provider will want to structure their contract with the firm in such a way that the firm is incentivized to invest in the cybersecurity of the medical data. 

In this scenario that we consider, the health care provider has two options available to mitigate the liability risks associated with a data breach. The first is that the health care provider is able to impose a fine or penalty on the firm if the technology firm suffers from a data breach. If this fine or penalty is sufficiently large, it can incentivize the firm to invest in cybersecurity that protects the data. This fine or penalty is in addition to the fee that the firm is charged in return for access to the medical data. The second is that the health care provider is able to purchase cybersecurity insurance from an external insurance agency.  

\subsection{Technology Firm Model}

The financial value to the technology firm of the shared medical data is $V$. This financial value is derived from the firm's ability to use the data to develop new artificial intelligence algorithms for health care, which can be sold to various health care providers. To get this data, the firm must pay the quantity $\phi$ to the health care provider. The technology firm is responsible for securing the data they receive. If the firm suffers from a data breach, they are required by the contract to pay a fine or penalty $t$ to the health care provider. 

The technology firm chooses an investment level $i$ in cybersecurity that protects the data. The firm chooses between a high ($i = 1$) or low ($i = 0$) level of investment. If the firm chooses high investment, then the probability of a breach is $\alpha \in (0,1)$, and the firm spends $\psi$ for this investment level. If the firm chooses low investment, then the probability of a breach is $\gamma\in(0,1)$, and (without loss of generality) the firm has zero expenditure for this investment level. We assume that $\alpha < \gamma$, meaning that a high level of investment strictly \emph{lowers} the probability of a data breach. We assume that the technology firm chooses their investment level by maximizing their expected profit:
\begin{equation}
    i^*(\phi, t) = \arg\max_{i\in\{0,1\}} (1-i)\cdot F^l(\phi,t) + i\cdot F^h(\phi,t)
\end{equation}
where expected profit under a low level of investment is 
\begin{equation}
    F^l(\phi,t) = V - \phi - \gamma\cdot t,
\end{equation}
and the expected profit under a high level of investment is 
\begin{equation}
    F^h(\phi,t) = V - \phi - \psi - \alpha\cdot t.
\end{equation}
Note that the technology firm is \emph{risk neutral} in this model.

\subsection{Health Care Provider Model}

The financial value to the health care provider of their own medical data is $W$. This financial value is derived from the provider's ability to use the data to self-improve the quality of its health care services through improved patient treatment and care delivery processes, as well as through medical research. If the technology firm suffers a data breach, then the health care provider has to spend $L$ to address its various liabilities to the affected patients. Since $t$ is a fine or penalty on the firm in the event of a breach, we assume $t \leq L$. Having $t > L$ is unrealistic because it would mean the healthcare provider profits from a data breach at the firm.

Furthermore, the health care provider can choose to purchase a policy to insure against their liabilities in the event of a breach. Under the assumption of an actuarially fair policy \cite{schwartz2014}, which would be expected to occur when there are a large number of insurers in the insurance marketplace, the health care provider can purchase an insurance policy that pays out $L_c$ under the event of a data breach at the cost of $p L_c$, where $p$ is the probability of a data breach. 

Finally, we assume the health care provider is \emph{risk averse}. This means that if the health care provider earns a financial revenue of $x$, then their utility for that revenue is $U(x)$ for a function $U(\cdot)$ that is strictly increasing and concave. Under the additional assumption that $U(\cdot)$ is differentiable, this risk aversion assumption means that $U'(\cdot) > 0$ and $U''(\cdot) < 0$.

\subsection{Contract Design Problem}

In this scenario, the health care provider faces a contract design problem in which their goal is to pick the purchase price $\phi$, the value of the fine or penalty $t$, and the insurance policy payout $L_c$ so as to maximize their own expected utility. This contract design problem can be written as the following bilevel program:
\begin{equation}
\label{eqn:s1cdp}
    \begin{aligned}
        \max_{\phi,t,L_c}\ & (1-i^*(\phi,t))\cdot H^l(\phi,t,L_c) + i^*(\phi,t)\cdot H^h(\phi,t,L_c)\\
        \mathrm{s.t.}\ & i^*(\phi, t) = \arg\max_{i\in\{0,1\}} (1-i)\cdot F^l(\phi,t) + i\cdot F^h(\phi,t)\\
        &(1-i^*(\phi, t))\cdot F^l(\phi,t) + i^*(\phi, t)\cdot F^h(\phi,t) \geq 0\\
        & \phi \geq 0, \ t\in[0,L], \ L_c \geq 0 
    \end{aligned}
\end{equation}
where we note that the health care provider's expected utility when the technology firm has a low level of investment in cybersecurity of the health data is given by
\begin{multline}
    H^l(\phi,t,L_c) = \gamma\cdot U(W + \phi - \gamma\cdot L_c - L + t + L_c) + \\
    (1-\gamma)\cdot U(W + \phi - \gamma\cdot L_c),
\end{multline}
the health care provider's expected utility when the technology firm has a high level of investment is 
\begin{multline}
    H^h(\phi,t,L_c) = \alpha\cdot U(W + \phi - \alpha\cdot L_c - L + t + L_c) + \\
    (1-\alpha)\cdot U(W + \phi - \alpha\cdot L_c),
\end{multline}
and the second constraint in (\ref{eqn:s1cdp}) is a \emph{participation constraint} that ensures the purchase cost $\phi$ and fine or penalty $t$ are such that the technology firm does not expect to lose money.

\subsection{Optimal Contract}

We next proceed to solve the contract design problem (\ref{eqn:s1cdp}) through a series of steps. Let $\phi^*,t^*,L_c^*$ denote the optimal contract, meaning they maximize (\ref{eqn:s1cdp}). We first characterize the optimal insurance coverage pay out.

\begin{proposition}
We have that $L_c^* = L - t^*$.
\end{proposition}

\begin{proof}
We first consider the case where $i^*(\phi^*,t^*) = 0$. In this case, the objective function of (\ref{eqn:s1cdp}) is $H^l(\phi,t,L_c)$, and the second constraint in (\ref{eqn:s1cdp}) is $F^l(\phi,t) \geq 0$. Next, note that the first-order stationarity condition is 
\begin{multline}
    0 = \partial_{L_c} H^l(\phi,t,L_c) = \\
    \gamma\cdot (1-\gamma)\cdot U'(W + \phi - \gamma\cdot L_c - L + t + L_c) +\\
    -\gamma\cdot (1-\gamma)\cdot U'(W + \phi - \gamma\cdot L_c).
\end{multline}
Since we assumed that $U''(\cdot) < 0$, this means the above is satisfied when $W + \phi - \gamma\cdot L_c - L + t + L_c = W + \phi - \gamma\cdot L_c$. Hence at optimality we have $L_c^* = L - t^*$, which is feasible since $t^* < L$ implies $L_c^* \geq 0$. The proof for the case $i^*(\phi^*,t^*) = 1$ proceeds almost identically.
\end{proof}

The implication of the above result is that we can rewrite the contract design problem as
\begin{equation}
\label{eqn:s1cdp_2}
    \begin{aligned}
        \max_{\phi,t}\ & (1-i^*(\phi,t))\cdot H^l(\phi,t) + i^*(\phi,t)\cdot H^h(\phi,t)\\
        \mathrm{s.t.}\ & i^*(\phi, t) = \arg\max_{i\in\{0,1\}} (1-i)\cdot F^l(\phi,t) + i\cdot F^h(\phi,t)\\
        &(1-i^*(\phi, t))\cdot F^l(\phi,t) + i^*(\phi, t)\cdot F^h(\phi,t) \geq 0\\
        & \phi \geq 0, \ t\in[0,L]
    \end{aligned}
\end{equation}
where
\begin{equation}
    H^l(\phi,t) := H^l(\phi,t,L-t) = U(W + \phi - \gamma\cdot(L - t)),
\end{equation}
and 
\begin{equation}
    H^h(\phi,t) := H^h(\phi,t,L-t) = U(W + \phi - \alpha\cdot(L - t)).
\end{equation}
We next use the above reformulation to characterize the optimal purchase price and fine or penalty amount.

\begin{proposition}
\label{prop:tstar}
If $i^*(\phi^*,t^*) = 0$, then an optimal choice solves the optimization problem
\begin{equation}
\label{eqn:tstar1}
\begin{aligned}
    \max_{\phi,t}\ &\phi + \gamma\cdot t\\
    \mathrm{s.t. }\ & \phi + \gamma\cdot t \leq V \\
    &(\gamma-\alpha)\cdot t < \psi\\
    &\phi \geq 0, \ t\in[0,L]
\end{aligned}
\end{equation}
If $i^*(\phi^*,t^*) = 1$, then an optimal choice solves the optimization problem
\begin{equation}
\label{eqn:tstar2}
\begin{aligned}
    \max_{\phi,t}\ &\phi + \alpha\cdot t\\
    \mathrm{s.t. }\ & \phi + \alpha\cdot t \leq V-\psi \\
    &(\gamma-\alpha)\cdot t \geq \psi\\
    &\phi \geq 0, \ t\in[0,L]
\end{aligned}
\end{equation}
\end{proposition}

\begin{proof}
We first consider the case where $i^*(\phi^*,t^*) = 0$. In this case, we have that $F^l(\phi,t) > F^h(\phi,t)$ (which is equivalent to $\psi  > (\gamma-\alpha)\cdot t$), that the second constraint of (\ref{eqn:s1cdp_2}) is 
\begin{equation}
    F^l(\phi,t)  = V - \phi - \gamma\cdot t\geq 0,
\end{equation} 
(which is equivalent to $\phi +\gamma\cdot t \leq V)$, and that the objective function of (\ref{eqn:s1cdp_2}) is $H^l(\phi,t)$. Since $U'(\cdot) > 0$, this means $H^l(\phi,t)$ is strictly increasing in $\phi +\gamma\cdot t$. The above observations imply that (\ref{eqn:tstar1}) provides an optimal choice. The proof for the case $i^*(\phi^*,t^*) = 1$ is nearly identical.
\end{proof}

We conclude by using the above characterization to finish deriving an optimal contract for this scenario.

\begin{theorem}
\label{thm:con1}
If $\psi > (\gamma-\alpha)\cdot L$ or $\psi > (\gamma-\alpha)\cdot V/\gamma$, then an optimal contract is $(\phi^*,t^*,L_c^*) = (V, 0, L)$. If $\psi \leq (\gamma-\alpha)\cdot L$ and $\psi \leq (\gamma-\alpha)\cdot V/\gamma$, then an optimal contract is $(\phi^*,t^*,L_c^*) = (V - \gamma/(\gamma-\alpha)\psi, \psi/(\gamma-\alpha), L - \psi/(\gamma-\alpha))$.
\end{theorem}

\begin{proof}
If $\psi > (\gamma-\alpha)\cdot L$, then this means $W + V -\gamma\cdot L > W + V - \psi - \alpha\cdot L$, which implies $H^l > H^h$ since $U'(\cdot) > 0$. Hence, the optimal choice is $i^*(\phi,t) = 0$. This means the choice $t^* = 0$ and $\phi^* = V$ is optimal by Proposition \ref{prop:tstar}. If $\psi > (\gamma-\alpha)\cdot V/\gamma$, then (\ref{eqn:tstar2}) is infeasible. This means applying Proposition \ref{prop:tstar} tells us that the optimal choice is $i^*(\phi,t) = 0$, and that we can again choose $t^* = 0$ and $\phi^* = V$. Finally, if $\psi \leq (\gamma-\alpha)\cdot L$ and $\psi \leq (\gamma-\alpha)\cdot V/\gamma$, then this means $W + V -\gamma\cdot L \leq W + V - \psi - \alpha\cdot L$, which implies $H^l \leq H^h$ since $U'(\cdot) > 0$. Moreover, the condition $\psi \leq (\gamma-\alpha)\cdot V/\gamma$ implies that (\ref{eqn:tstar2}) is feasible. This means the choice $t^* = \psi/(\gamma-\alpha)$ (and note $t^* \leq L$ since $\psi/(\gamma-\alpha) \leq L$ in this case) and $\phi^* = V-\psi - \alpha/(\gamma-\alpha)\psi = V - \gamma/(\gamma-\alpha)\psi$ provides an optimal contract for this case. 
\end{proof}

\subsection{Insights from the Optimal Contract}

Several insights can be gained from the optimal contract in Theorem \ref{thm:con1}. The most interesting insights are related to the conditions that result in a contract where the technology firm makes a high or low investment in cybersecurity: 

When $\psi > (\gamma-\alpha)\cdot L$ or $\psi > (\gamma-\alpha)\cdot V/\gamma$, the optimal contract leads to the technology firm making a low investment in cybersecurity. If $\psi > (\gamma-\alpha)\cdot V/\gamma$, then a high investment $\psi$ by the technology firm in cybersecurity is relatively costly compared to the financial value $V$ to the \emph{technology firm} of the medical data, and the technology firm will not want to make a high investment in cybersecurity. If $\psi > (\gamma-\alpha)\cdot L$, then this means that a a high investment $\psi$ by the technology firm in cybersecurity is relatively costly compared to the liability $L$ of the \emph{health care provider} in the event of a data breach. It is surprising that the optimal contract when $\psi > (\gamma-\alpha)\cdot L$ holds also leads to the technology firm making a low investment in cybersecurity.

Another interesting aspect of the optimal contract when $\psi > (\gamma-\alpha)\cdot L$ or $\psi > (\gamma-\alpha)\cdot V/\gamma$ is that the optimal contract is such that there is no penalty or fine $t^* = 0$ to the technology firm in the event of a data breach. In this case, it is instead optimal to charge as much as possible for the data, meaning it is optimal to charge $\phi^* = V$.

On the other hand, the optimal contract induces the technology firm to invest in cybersecurity only when $\psi \leq (\gamma-\alpha)\cdot L$ and $\psi \leq (\gamma-\alpha)\cdot V/\gamma$. A numerical example of these thresholds are shown in Fig. \ref{fig:fig1}. Here, a high investment $\psi$ by the technology firm in cybersecurity is relatively cheap compared to the financial value $V$ to the \emph{technology firm} of the medical data and relatively small compared to the liability $L$ of the \emph{health care provider} in the event of a data breach.

\begin{figure}
    \centering
    \includegraphics{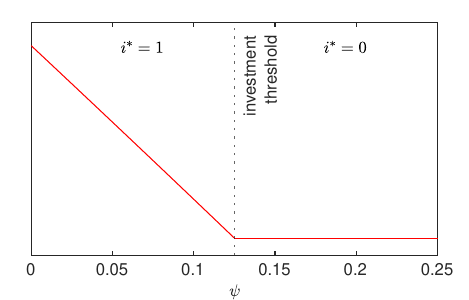}
    \caption{This plot shows the expected utility from an optimal contract as a function of $\psi$. We can observe that there exists a threshold beyond which the cost of investment is relatively high compared to either the financial value of the medical data or the liability in the event of a data breach, and thus it is optimal to induce $i^* = 0$ if $\psi$ is above this threshold value.}
    \label{fig:fig1}
\end{figure}

\section{Scenario B: Consortium of Health Care Providers Sharing Data}
\label{sec:scenariob}

The second scenario we study is that of a group of health care providers sharing medical data amongst themselves. We assume that the cybersecurity level of each health care provider is fixed, and that all health care providers have identical models. Here, the primary consideration is each health care provider's  decision of whether or not to join the consortium. If any single health care provider in the consortium suffers from a data breach, then all the health care providers in the consortium will face liability from those patients whose medical data has been breached. Thus, the health care providers must decide whether any benefits accrued from being in the consortium outweigh the increased risks of data breach due to sharing medical data. 

In this scenario that we consider, each health care provider has two options available to mitigate the liability risks associated with a data breach. The first is that consortium can impose a fine or penalty on the health care provider that suffers from a data breach, which is equally shared by the remaining health care providers. The second is that each health care provider is able to purchase cybersecurity insurance from an external insurance agency.

\subsection{Health Care Provider Model}

The financial value to a health care provider of their own medical data is $W$, and the financial value to a health care provider of the medical data from a consortium of $k$ health care providers is $v(k)\cdot W$, where the function $v(\cdot) > 0$ is strictly increasing and concave with $v(1) = 1$. This financial value is derived from the provider's ability to use the data to self-improve its health care services through improved patient treatment and care delivery processes, as well as through medical research. This model says that more quantity of data gives more financial value, but that there are diminishing financial returns to increasing quantities of data. 

If any consortium member suffers a data breach, then each health care provider in the consoritum has to spend $L$ to address its various liabilities to the affected patients. The probability of a data breach among $k$ health care providers is given by $p(k) \in (0,1)$. We assume that this function $p(\cdot)$ is concave and increases with a sublinear growth rate such that $p(k) < k\cdot p(1)$. Furthermore, the health care provider responsible for the data breach is required to pay a fine or penalty of $(k-1)\cdot t$ that is equally divided among the $(k-1)$ remaining health care providers, where we assume $t \leq L$. Having $t > L$ is unrealistic because it would mean a healthcare provider profits from a data breach elsewhere. 

Each health care provider can choose to purchase a policy to insure against their liabilities in the event of a breach. Under the assumption of an actuarially fair policy \cite{schwartz2014}, each health care provider can purchase an insurance policy that pays out $L_c$ under the event of a data breach at the cost of $p L_c$, where $p$ is the probability of a data breach. 

Finally, we assume each health care provider is \emph{risk averse}. This means that if a health care provider earns a financial revenue of $x$, then their utility for that revenue is $U(x)$ for a function $U(\cdot)$ that is strictly increasing and concave. Under the additional assumption that $U(\cdot)$ is differentiable, this assumption of risk aversion means that $U'(\cdot) > 0$ and $U''(\cdot) < 0$. 

\subsection{Contract Design Problem}

In this scenario, the consortium faces a contract design problem in which their goal is to pick the value of the fine or penalty $t$ and the insurance policy payout $L_c$ so as to encourage participation in the consortium and thus motivate data sharing. This contract design problem can be written as the following:
\begin{equation}
\label{eqn:s2cdp}
    \begin{aligned}
        \max_{s,t,L_c}\ & (1-s)\cdot H^1(L_c) + s\cdot H^k(t,L_c)\\
        \mathrm{s.t.}\         & t\in[0,L], \ L_c \geq 0 
    \end{aligned}
\end{equation}
where we note that a health care provider's expected utility when they do not participate in the consortium is given by 
\begin{multline}\label{eqn:h1lc}
    H^1(L_c) = p(1)\cdot U(W - p(1)\cdot L_c - L + L_c) + \\
    (1-p(1))\cdot U(W - p(1)\cdot L_c),
\end{multline}
and a health care provider's expected utility when they do participate in the consortium is given by
\begin{multline}
    H^k(t,L_c) = \\
    p(1)\cdot U(v(k)\cdot W - p(k)\cdot L_c - L - (k-1)\cdot t + L_c) + \\
    (p(k)-p(1))\cdot U(v(k)\cdot W - p(k)\cdot L_c - L + t + L_c) + \\
    (1-p(k))\cdot U(v(k)\cdot W - p(k)\cdot L_c).
\end{multline}

\subsection{Optimal Contract}

We next proceed to solve the contract design problem (\ref{eqn:s2cdp}) through a series of steps. Let $s^*,t^*,L_c^*$ denote the optimal contract, meaning they maximize (\ref{eqn:s2cdp}). We first characterize the optimal fine or penalty amount.

\begin{proposition}
\label{prop:tstar_2}
We have that $t^* = 0$.
\end{proposition}

\begin{proof}
If $s^* = 0$, then the objective function of (\ref{eqn:s2cdp}) is $H^1(L_c)$. This means the objective function value does not depend on $t$, and so any feasible $t$ is optimal. Hence, we can pick $t^* = 0$ in this case. If $s^* = 1$, then the objective function of (\ref{eqn:s2cdp}) is $H^k(t,L_c)$. Now, consider the partial derivative with respect to $t$:
\begin{multline}
    \partial_{t} H^k(t,L_c) = \\
    -(k-1)\cdot p(1)\cdot U'(v(k)\cdot W - p(k)\cdot L_c - L - (k-1)\cdot t + L_c) + \\
    (p(k)-p(1))\cdot U'(v(k)\cdot W - p(k)\cdot L_c - L + t + L_c).
\end{multline}
Since we assumed that $U''(\cdot) < 0$, this means $U'(v(k)\cdot W - p(k)\cdot L_c - L - (k-1)\cdot t + L_c) > U'(v(k)\cdot W - p(k)\cdot L_c - L + t + L_c)$ since $v(k)\cdot W - p(k)\cdot L_c - L - (k-1)\cdot t + L_c < v(k)\cdot W - p(k)\cdot L_c - L + t + L_c$. Recalling that $p(k) > p(1)$ and $U'(\cdot) > 0$ by assumption, we thus have 
\begin{multline}
\partial_{t} H^k(t,L_c) \leq
    \big[-(k-1)\cdot p(1) + (p(k)-p(1))\big]\times\\
    U'(v(k)\cdot W - p(k)\cdot L_c - L - (k-1)\cdot t + L_c).
\end{multline}
Since we assumed that $p(k) < k\cdot p(1)$ and $U'(\cdot) > 0$, this means $\partial_{t} H^k(t,L_c) < 0$. Thus choosing $t^* = 0$ is optimal because we are constrained in (\ref{eqn:s2cdp}) to choose $t\in[0,1]$.
\end{proof}

The implication of the above result is that we can rewrite the contract design problem as
\begin{equation}
\label{eqn:s2cdp_2}
    \begin{aligned}
        \max_{s,L_c}\ & (1-s)\cdot H^1(L_c) + s\cdot H^k(L_c)\\
        \mathrm{s.t.}\         & L_c \geq 0 
    \end{aligned}
\end{equation}
where $H^1(\cdot)$ is as defined in (\ref{eqn:h1lc}), and 
\begin{multline}
    H^k(L_c) := H^k(0,L_c) = \\
    p(k)\cdot U(v(k)\cdot W - p(k)\cdot L_c - L + L_c) + \\
    (1-p(k))\cdot U(v(k)\cdot W - p(k)\cdot L_c).
\end{multline}
We next use the above reformulation to characterize the optimal insurance coverage payout.

\begin{proposition}
\label{prop:lcstar_2}
We have that $L_c^* = L$.
\end{proposition}

\begin{proof}
If $s^* = 0$, then the objective function of (\ref{eqn:s2cdp_2}) is $H^1(L_c)$. Next, note the first-order stationarity condition is 
\begin{multline}
    0 = \partial_{L_c} H^1(L_c) = \\
    p(1)\cdot (1-p(1))\cdot U'(W - p(1)\cdot L_c - L + L_c) +\\
    -p(1)\cdot (1-p(1))\cdot U'(W - p(1)\cdot L_c).
\end{multline}
Since we assumed that $U''(\cdot) < 0$, this means the above is satisfied when $W - p(1)\cdot L_c - L + L_c = W - p(1)\cdot L_c$. Hence at optimality we have $L_c^* = L$, which is feasible since $L >0$ implies $L_c^* \geq 0$. The proof for the case $s^* = 1$ proceeds almost identically.
\end{proof}


The implication of the above result is that we can rewrite the contract design problem as
\begin{equation}
\label{eqn:s2cdp_3}
    \begin{aligned}
        \max_{s}\ & (1-s)\cdot H^1 + s\cdot H^k\\
    \end{aligned}
\end{equation}
where
\begin{equation}
    H^1 := H^1(L) = U(W - p(1)\cdot L)
\end{equation}
and
\begin{equation}
    H^k := H^k(L)  = U(v(k)\cdot W - p(k)\cdot L).
\end{equation}
We conclude by using the above characterization to finish deriving an optimal contract for this scenario.

\begin{theorem}
\label{thm:con2}
If $W - p(1)\cdot L > v(k)\cdot W - p(k)\cdot L$, then an optimal contract is given by $(s^*,t^*,L_c^*) = (0, 0, L)$. If $W - p(1)\cdot L \leq v(k)\cdot W - p(k)\cdot L$, then an optimal contract is given by $(s^*,t^*,L_c^*) = (1, 0, L)$.
\end{theorem}

\begin{proof}
If $W - p(1)\cdot L > v(k)\cdot W - p(k)\cdot L$, then this means $H^1 = U(W - p(1)\cdot L) > H^k = U(v(k)\cdot W - p(k)\cdot L)$ since $U'(\cdot) > 0$. Hence, the optimal choice is $s^* = 0$. This means the choice $t^* = 0$ and $L_c^* = L$ is optimal by Propositions \ref{prop:tstar_2} and \ref{prop:lcstar_2}. If $W - p(1)\cdot L \leq v(k)\cdot W - p(k)\cdot L$, then this means $H^1 = U(W - p(1)\cdot L) \leq H^k = U(v(k)\cdot W - p(k)\cdot L)$ since $U'(\cdot) > 0$. Hence, the optimal choice is $s^* = 1$. This means the choice $t^* = 0$ and $L_c^* = L$ is optimal by Propositions \ref{prop:tstar_2} and \ref{prop:lcstar_2}.
\end{proof}

\subsection{Insights from the Optimal Contract}

Several insights can be gained from the optimal contract in Theorem \ref{thm:con2}. One interesting insight is that the optimal contract has $t^* = 0$, meaning that there is no penalty or fine in the event of a data breach, even when the health care providers participate in the consortium. This has an important practical implication, which is that participation in the data sharing consortium can only be encouraged by the ability of a health care provider to purchase an insurance policy from an external insurance company. Specifically, the optimal contract has $L_c^* = L$. This means $H^k = H^k(L) > H^k(0)$, or in words that purchasing insurance gives each participating health care provider a strictly higher utility than \emph{not} purchasing insurance. Restated, the ability to purchase insurance for the event of a data breach makes it more likely for a health care provider to be willing to share data. Furthermore, Fig. \ref{fig:fig2} shows that depending on the particular functional forms, there is often a maximum consortium size beyond which costs associated with the increased likelihood of data breaches exceeds the value of sharing more data.

\begin{figure}
    \centering
    \includegraphics{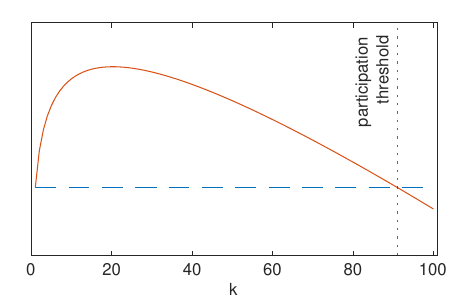}
    \caption{The red solid line is a health care provider's expected utility when they participate in the consortium, and the blue dashed line is the expected utility when they do not participate, where $k$ is the number of health care providers in the consortium. The black dotted line represents the participation threshold beyond which the risks outweigh the benefits of sharing data by participating in the consortium.}
    \label{fig:fig2}
\end{figure}

\section{Conclusion}
\label{sec:conclusion}

In this paper, we designed contracts that help to mitigate the risks associated with data sharing so as to encourage health care providers to share their medical data. We first studied a scenario where a single health care provider sells medical data to a technology firm that is interested in using the data to develop new artificial intelligence algorithms. We next studied a scenario where multiple health care providers share data with each other for the purpose of conducting scientific research and improving patient care. Both cases required managing a trade-off between the value of sharing data with the liabilities associated with data breaches. The key concepts towards managing the risks associated with data breaches were the ideas of imposing a fine or penalty and purchasing external insurance to mitigate liabilities in the event of a data breach. Our results suggest that it is possible to devise contracts that promote the sharing of medical data while preserving the integrity and privacy of the data. By implementing the correct incentives, it may be possible to overcome the barriers to data sharing and facilitate the use of health information for science, technology, and policy. \\




\bibliographystyle{IEEEtran}
\bibliography{cybersecurity}

\begin{thebibliography}{10}
\providecommand{\url}[1]{#1}
\csname url@samestyle\endcsname
\providecommand{\newblock}{\relax}
\providecommand{\bibinfo}[2]{#2}
\providecommand{\BIBentrySTDinterwordspacing}{\spaceskip=0pt\relax}
\providecommand{\BIBentryALTinterwordstretchfactor}{4}
\providecommand{\BIBentryALTinterwordspacing}{\spaceskip=\fontdimen2\font plus
\BIBentryALTinterwordstretchfactor\fontdimen3\font minus
  \fontdimen4\font\relax}
\providecommand{\BIBforeignlanguage}[2]{{%
\expandafter\ifx\csname l@#1\endcsname\relax
\typeout{** WARNING: IEEEtran.bst: No hyphenation pattern has been}%
\typeout{** loaded for the language `#1'. Using the pattern for}%
\typeout{** the default language instead.}%
\else
\language=\csname l@#1\endcsname
\fi
#2}}
\providecommand{\BIBdecl}{\relax}
\BIBdecl

\bibitem{roski2014}
J.~Roski, G.~W. Bo-Linn, and T.~A. Andrews, ``Creating value in health care
  through big data: Opportunities and policy implications,'' \emph{Health
  Affairs}, vol.~33, no.~7, pp. 1115--1122, 2014.

\bibitem{lee2017}
C.~H. Lee and H.-J. Yoon, ``Medical big data: promise and challenges,''
  \emph{Kidney Research and Clinical Practice}, vol.~36, no.~1, pp. 3--11,
  2017.

\bibitem{mintz2017behavioral}
Y.~Mintz, A.~Aswani, P.~Kaminsky, E.~Flowers, and Y.~Fukuoka, ``Behavioral
  analytics for myopic agents,'' \emph{arXiv preprint arXiv:1702.05496}, 2017.

\bibitem{zhou2018evaluating}
M.~Zhou, Y.~Fukuoka, Y.~Mintz, K.~Goldberg, P.~Kaminsky, E.~Flowers, and
  A.~Aswani, ``Evaluating machine learning--based automated personalized daily
  step goals delivered through a mobile phone app: Randomized controlled
  trial,'' \emph{JMIR mHealth and uHealth}, vol.~6, no.~1, p. e9117, 2018.

\bibitem{cdc2018}
{Office of Public Health Scientific Services}, ``Public health surveillance:
  Preparing for the future,'' Centers for Disease Control and Prevention, Tech.
  Rep., 2018.

\bibitem{aha2019}
{American Hospital Association}, ``Sharing data, saving lives: The hospital
  agenda for interoperability,'' American Hospital Association, Tech. Rep.,
  2019.

\bibitem{hulsen2020}
T.~Hulsen, ``Sharing is caring—data sharing initiatives in healthcare,''
  \emph{International Journal of Environmental Research and Public Health},
  vol.~17, no.~9, 2020.

\bibitem{kanonymity}
L.~Sweeney, ``K-anonymity: A model for protecting privacy,''
  \emph{International Journal of Uncertainty, Fuzziness and Knowledge-Based
  Systems}, vol.~10, no.~5, pp. 557--570, 2002.

\bibitem{ldiversity}
A.~Machanavajjhala, J.~Gehrke, D.~Kifer, and M.~Venkitasubramaniam,
  ``L-diversity: privacy beyond k-anonymity,'' in \emph{2006 IEEE 22nd
  International Conference on Data Engineering}, 2006, pp. 24--24.

\bibitem{tcloseness}
N.~Li, T.~Li, and S.~Venkatasubramanian, ``t-closeness: Privacy beyond
  k-anonymity and l-diversity,'' in \emph{2007 IEEE 23rd International
  Conference on Data Engineering}, 2007, pp. 106--115.

\bibitem{dprivacy}
C.~Dwork, ``Differential privacy,'' in \emph{Automata, Languages and
  Programming}, M.~Bugliesi, B.~Preneel, V.~Sassone, and I.~Wegener, Eds.\hskip
  1em plus 0.5em minus 0.4em\relax Berlin, Heidelberg: Springer, 2006, pp.
  1--12.

\bibitem{hipaa}
\BIBentryALTinterwordspacing
{U.S. Department of Health and Human Services}, ``{The Health Insurance
  Portability and Accountability Act Privacy Rule},'' (Accessed: Mar. 31,
  2022). [Online]. Available:
  \url{https://www.hhs.gov/hipaa/for-professionals/privacy/index.html}
\BIBentrySTDinterwordspacing

\bibitem{gdpr}
\BIBentryALTinterwordspacing
{European Union}, ``{Regulation (EU) 2016/679 of the European Parliament and of
  the Council of 27 April 2016 on the protection of natural persons with regard
  to the processing of personal data and on the free movement of such data, and
  repealing Directive 95/46/EC (General Data Protection Regulation)},''
  (Accessed: Mar. 31, 2022). [Online]. Available:
  \url{https://eur-lex.europa.eu/eli/reg/2016/679/oj}
\BIBentrySTDinterwordspacing

\bibitem{narayanan2008}
A.~Narayanan and V.~Shmatikov, ``Robust de-anonymization of large sparse
  datasets,'' in \emph{2008 IEEE Symposium on Security and Privacy}, 2008, pp.
  111--125.

\bibitem{ohm2010}
P.~Ohm, ``Broken promises of privacy: Responding to the surprising failure of
  anonymization,'' \emph{UCLA Law Review}, vol.~57, pp. 1701--1777, 2010.

\bibitem{datta2012}
A.~Datta, D.~Sharma, and A.~Sinha, ``Provable de-anonymization of large
  datasets with sparse dimensions,'' in \emph{International Conference on
  Principles of Security and Trust}.\hskip 1em plus 0.5em minus 0.4em\relax
  Springer, 2012, pp. 229--248.

\bibitem{oxford2022}
E.~Oxford, ``Hundreds of patient data breaches are left unpunished,''
  \emph{BMJ}, vol. 377, 2022.

\bibitem{sweeney1997}
L.~Sweeney, ``Weaving technology and policy together to maintain
  confidentiality,'' \emph{The Journal of Law, Medicine \& Ethics}, vol.~25,
  no. 2-3, pp. 98--110, 1997.

\bibitem{na2018}
L.~Na, C.~Yang, C.-C. Lo, F.~Zhao, Y.~Fukuoka, and A.~Aswani, ``Feasibility of
  reidentifying individuals in large national physical activity data sets from
  which protected health information has been removed with use of machine
  learning,'' \emph{JAMA Network Open}, vol.~1, no.~8, 2018.

\bibitem{jiang2020}
J.~Jiang and G.~Bai, ``Types of information compromised in breaches of
  protected health information,'' \emph{Annals of Internal Medicine}, vol. 172,
  no.~2, pp. 159--160, 2020.

\bibitem{panhuis2014}
W.~{Van Panhuis}, P.~Paul, C.~Emerson, J.~Grefenstette, R.~Wilder, A.~Herbst,
  D.~Heymann, and D.~Burke, ``\BIBforeignlanguage{English (US)}{A systematic
  review of barriers to data sharing in public health},''
  \emph{\BIBforeignlanguage{English (US)}{BMC Public Health}}, vol.~14, no.~1,
  2014.

\bibitem{pisani2010}
E.~Pisani and C.~AbouZahr, ``Sharing health data: good intentions are not
  enough,'' \emph{Bulletin of the World Health Organization}, vol.~88, no.~6,
  pp. 462--466, 2010.

\bibitem{marotta2017}
A.~Marotta, F.~Martinelli, S.~Nanni, A.~Orlando, and A.~Yautsiukhin,
  ``Cyber-insurance survey,'' \emph{Computer Science Review}, vol.~24, pp.
  35--61, 2017.

\bibitem{kunreuther2003}
H.~Kunreuther and G.~Heal, ``Interdependent security,'' \emph{Journal of Risk
  and Uncertainty}, vol.~26, no.~2, pp. 231--249, 2003.

\bibitem{ogut2005}
H.~Ogut, N.~M. Menon, and S.~Raghunathan, ``Cyber insurance and it security
  investment: Impact of interdependence risk,'' in \emph{WEIS}, 2005.

\bibitem{bolot2008}
J.~C. Bolot and M.~Lelarge, ``A new perspective on internet security using
  insurance,'' in \emph{IEEE INFOCOM 2008}, 2008, pp. 1948--1956.

\bibitem{bohme2005}
R.~B{\"o}hme, ``Cyber-insurance revisited,'' in \emph{WEIS}, 2005.

\bibitem{bohme2006}
R.~B{\"o}hme and G.~Kataria, ``Models and measures for correlation in
  cyber-insurance,'' in \emph{WEIS}, 2006.

\bibitem{shetty2010}
N.~Shetty, G.~Schwartz, M.~Felegyhazi, and J.~Walrand, ``Competitive
  cyber-insurance and internet security,'' in \emph{Economics of Information
  Security and Privacy}, T.~Moore, D.~Pym, and C.~Ioannidis, Eds.\hskip 1em
  plus 0.5em minus 0.4em\relax Boston, MA: Springer US, 2010, pp. 229--247.

\bibitem{bohme2010}
R.~B{\"o}hme and G.~Schwartz, ``Modeling cyber-insurance: Towards a unifying
  framework,'' in \emph{WEIS}, 2010.

\bibitem{schwartz2014}
G.~A. Schwartz and S.~S. Sastry, ``Cyber-insurance framework for large scale
  interdependent networks,'' in \emph{HiCoNS}, 2014, p. 145–154.

\bibitem{laffont2002}
J.-J. Laffont and D.~Martimort, \emph{The Theory of Incentives: The
  Principal-Agent Model}.\hskip 1em plus 0.5em minus 0.4em\relax Princeton
  University Press, 2002.

\end{thebibliography}

\end{document}